\documentclass[abstracton, DIV=14, 11pt,letterpaper]{article}
\usepackage[margin=1in]{geometry}
\usepackage{times}  % DO NOT CHANGE THIS
\usepackage{helvet} % DO NOT CHANGE THIS
\usepackage{courier}  % DO NOT CHANGE THIS
\usepackage[hyphens]{url}  % DO NOT CHANGE THIS
\usepackage{graphicx} % DO NOT CHANGE THIS
\usepackage{natbib}  % DO NOT CHANGE THIS AND DO NOT ADD ANY OPTIONS TO IT
\usepackage{caption} % DO NOT CHANGE THIS AND DO NOT ADD ANY OPTIONS TO IT
\usepackage{amsfonts}
\usepackage{amsthm}
\usepackage{amsmath}
\usepackage{xcolor}
\usepackage{authblk}
\usepackage{thmtools}
\usepackage{thm-restate}

\usepackage{graphicx}
\usepackage{comment}
\usepackage[ruled,linesnumbered,norelsize]{algorithm2e}
\usepackage[switch]{lineno}  %
\usepackage{color}
\definecolor{DarkBlue}{RGB}{0,0,150}
\usepackage[colorlinks,linkcolor=DarkBlue,citecolor=DarkBlue]{hyperref}

\theoremstyle{definition}

\newtheorem{definition}{Definition}

\newtheorem{corollary}{Corollary}
\newtheorem{proposition}{Proposition}
\newtheorem{lemma}{Lemma}

\def\argmax{\ensuremath{\mathrm{argmax}}}
\newcommand{\MM}{\text{SPM}}
\DeclareMathOperator{\E}{\mathbb{E}}
\newcommand{\D}{\mathcal{D}}
\newcommand{\vals}{\mathbf{v}}

\newcommand{\alloc}{\mathbf{x}}
\newcommand{\p}{\boldsymbol{\mathbf{\tau}}}

\usepackage{tikz}
\usetikzlibrary{matrix,shapes,arrows,positioning,chains}

\begin{document}
	
\title{Reinforcement Learning of Sequential Price Mechanisms\footnote{Author order is alphabetical. This is the extended version of \citet{brero2021rlspm}, which was presented at the Thirty-Fifth AAAI Conference on Artificial Intelligence (AAAI'21). This research is funded in part by Defense Advanced Research Projects Agency under Cooperative Agreement HR00111920029. The content of the information does not necessarily reflect the position or the policy of the Government, and no official endorsement should be inferred. This is approved for public release; distribution is unlimited. The work of G. Brero was also supported by the SNSF (Swiss National Science Foundation) under Fellowship P2ZHP1\_191253.}}

\author[a]{Gianluca Brero}
\author[a]{Alon Eden}
\author[a]{Matthias Gerstgrasser}
\author[a]{David C.~Parkes}
\author[a]{\\Duncan Rheingans-Yoo}

\affil[a]{Harvard University \authorcr \texttt{gbrero,aloneden,matthias,parkes@g.harvard.edu}, \texttt{d.rheingansyoo@gmail.com}}

\date{\today}

\maketitle

\begin{abstract}
We introduce the use of reinforcement learning for indirect mechanisms, working with the existing class of {\em sequential price mechanisms}, which generalizes both serial dictatorship and posted price mechanisms and essentially characterizes all strongly obviously strategyproof mechanisms. Learning an optimal mechanism within this
class forms a partially-observable Markov decision process. We provide rigorous conditions for when this class of mechanisms is more powerful than simpler static mechanisms, for sufficiency or insufficiency of observation statistics for learning, and for the necessity of complex (deep) policies. We show that our approach can learn optimal or near-optimal mechanisms in several experimental settings.
\end{abstract}

\section{Introduction}\label{sec:intro}
Over the last fifty years, a large body of research in microeconomics has introduced many different mechanisms for resource allocation. Despite the wide variety of available options, ``simple'' mechanisms such as \textit{posted price} and \textit{serial dictatorship} are often preferred for practical applications, including housing allocation \citep{abdulkadirouglu1998random},
online procurement \citep{badanidiyuru2012learning}, or allocation of medical appointments \citep{klaus2019serial}.

There has also been considerable interest in formalizing different notions of simplicity. \citet{li2017obviously} identifies mechanisms that are particularly simple from a strategic perspective, introducing the concept of \textit{obviously strategyproof mechanisms}. These are mechanisms in which it is obvious that an agent cannot profit by trying to game the system, as even the worst possible final outcome from behaving truthfully is at least as good as the best possible outcome from any other strategy. More recently,~\citet{pycia2019theory} introduce the still stronger concept of \textit{strongly obviously strategyproof} (SOSP) mechanisms, and show that this class is essentially equivalent to the \textit{sequential price mechanisms}, where agents are visited in turn and offered a choice from a menu (which may or may not include transfers). SOSP  mechanisms are ones in which an agent is not even required to consider her future (truthful) actions to understand that a mechanism is obviously strategyproof.

Despite being simple to use, designing optimal sequential price mechanisms can be a hard task, even when targeting common objectives, such as maximum welfare or maximum revenue. For example, in unit-demand settings with multiple items, the problem of computing prices that maximize expected revenue given discrete prior distributions on buyer values is NP-hard~\citep{chen2014complexity}. More recently, \citet{AgrawalSZ20} showed a similar result for the problem of determining an optimal order in which agents will be visited when selling a single item using posted price mechanisms.

\paragraph{Our Contribution.}\label{sec:contribution}
In this paper, we  introduce the first use of reinforcement learning (RL) for the design of indirect mechanisms, applying RL to the  design of optimal sequential price mechanisms ($\MM$s), and  demonstrate its effectiveness across a wide range of settings with different economic features.
We  generally  focus on mechanisms that optimize expected welfare. However, the framework is completely flexible, allowing for different objectives, and in addition to welfare, we  illustrate its use for max-min fairness and revenue.

The problem of learning an optimal $\MM$ is formulated as  a {\em partially observable Markov decision process} (POMDP). In this POMDP, the environment (i.e., the state, transitions, and rewards) models the economic setting, and the  policy, which observes   purchases and selects the next agent and prices based on those observations, encodes the mechanism rules. Solving for an optimal policy is equivalent to solving the mechanism design problem.
For the $\MM$ class, we can directly simulate agent behavior as part of the environment since there is a  dominant-strategy equilibrium.
We give requirements on the statistic of the history of observations needed to support an optimal policy and show that this statistic can be succinctly represented in the number of items and agents. {We also show that non-linear policies based on these statistics may be necessary to increase welfare. Accordingly, we use deep-RL algorithms to learn mechanisms.}

The theoretical results provide rigorous conditions for when $\MM$s are more powerful than simpler static mechanisms, providing a new understanding of this class of mechanisms. We show that for all but the simplest settings, adjusting the posted prices and the order in which agents are visited based on prior purchases improves welfare outcomes.
Lastly, we report on a comprehensive set of experimental results for the {\em Proximal Policy Optimization} (PPO) algorithm \citep{schulman2017proximal}.
We  consider a range of settings, from simple to more  intricate, that serve to illustrate  our theoretical results as well as generally demonstrate the performance of PPO, as well as the relative performance of $\MM$s in comparison to simple static mechanisms.

\paragraph{Further Related Work.}
\label{sec:literature}
Economic mechanisms based on sequential posted prices have been studied since the early 2000s. \citet{sandholm2003sequences} study {\em take-it-or-leave-it auctions} for a single item, visiting buyers in turn and making them offers. They introduced a linear-time algorithm that, in specific settings with two buyers, computes an optimal sequence of offers to maximize  revenue.
More recently, building on the prophet inequality literature, \citet{kleinberg2012matroid}, \citet{FeldmanGL15}, and \citet{dutting2016posted} derived different welfare and revenue guarantees for posted prices mechanisms for combinatorial auctions.
\citet{klaus2019serial} studied $\MM$s in settings with homogeneous items, showing that they satisfy many desirable properties in addition to being strategyproof.

Another related research thread is that of \textit{automated mechanism design} (AMD) \citep{ConitzerS02,conitzer2004self}, which seeks to use algorithms to design mechanisms.
Machine learning has been used for the design of direct mechanisms~\citep{dutting2015payment,narasimhan2016automated,Duetting0NPR19,GolowichNP18}, including sample complexity results~\cite[e.g]{ColeR14,GonczarowskiW18}.  There have also been important theoretical advances, identifying polynomial-time algorithms for direct-revelation, revenue-optimal mechanisms~\citep[e.g.]{cai2012algorithmic,cai2012optimal,cai2013understanding}.

Despite this rich research thread on direct mechanisms,
the use of AMD for indirect mechanisms is less well understood. Indirect mechanisms have an imperative nature (e.g., sequential, or multi-round), and may involve richer strategic behaviors. Machine learning has been used to realize indirect versions of mechanisms such as the VCG mechanism, or together with assumptions of truthful responses~\citep{lahaie2004applying,blum2004preference,brero2019machine}. Situated towards finding clearing prices for combinatorial auctions, the work by \citet{brero2019fast} involves inference about the  valuations of agents via Bayesian approaches.

Related to RL, but otherwise quite different from our setting,~\citet{ShenPLZQHGDLT20} study the design of reserve prices in repeated ad auctions, i.e., {\em direct} mechanisms, using an MDP framework to model the interaction between pricing and agent response across multiple instantiations of a mechanism (whereas, we  use a POMDP, enabling value inference across the rounds of a single $\MM$). This use of RL and MDPs for the design of repeated mechanisms has also been considered for  matching buyer impressions to sellers on platforms such as Taobao~\citep{Tang17a,CaiFTZ18}.

\section{Preliminaries}\label{sec:prelim}

\paragraph{Economic Framework.} \label{sec:model}

There are $n$ agents and $m$ indivisible items. Let $[n]=\{1,\ldots,n\}$ be the set of agents and $[m]$ be the set of items.
Agents have a  valuation function $v_i:2^{[m]}\rightarrow \mathbb{R}_{\ge 0}$ that maps bundles of items to a real value. As a special case, a {\em unit-demand valuation} is one in which  an agent has a value for each item, and the value for a bundle is the maximum value for an item in the bundle. Let $\vals=(v_1,\ldots,v_n)$  denote the valuation profile. We assume  $\vals$ is sampled from a possibly correlated value distribution $\D$. The designer can access this distribution $\D$  through samples from the  joint distribution.

An {\em allocation} $\alloc=(x_1,\ldots,x_n)$ is a profile of disjoint bundles of items ($x_i\cap x_j=\emptyset$ for every $i\neq j\in [n]$), where $x_i\subseteq [m]$ is the set of items allocated to agent $i$.

An {\em economic mechanism} interacts with agents and determines an outcome, i.e., an allocation $\alloc$ and transfers (payments) $\p=(\tau_1,\ldots,\tau_n)$, where $\tau_i \ge 0$ is the payment by agent $i$.
We measure the performance of a mechanism outcome $(\alloc, \p)$ under valuation profile $\vals$ via an objective function $\mathsf{g}(\alloc,\p;\vals)$.

\medskip
\noindent\textbf{Our goal:} Design a mechanism whose outcome maximizes the \textit{expected value} the objective function with respect to the value distribution.

\medskip
Our framework allows for different objectives such as:
\begin{itemize}
	\item social welfare: $\mathsf{g}(\alloc,\p;\vals)=\sum_{i\in [n]} v_i(x_i)$,
	\item revenue: $\mathsf{g}(\alloc,\p;\vals)=\sum_{i\in [n]} \tau_i$, and
	\item max-min fairness: $\mathsf{g}(\alloc,\p;\vals)=\min_{i\in [n]}v_i(x_i)$.
\end{itemize}

\paragraph{Sequential Price Mechanisms.}
We  study the family of $\MM$s. An $\MM$  interacts with agents across rounds, $t\in\{1,2,\ldots\}$, and  visits a different agent in each round. At the end of round $t$, the mechanism maintains the following parameters: a {\em temporary allocation} $\alloc^t$ of the first $t$ agents visited, a {\em temporary payment profile} $\p^t$, and a \textit{residual setting} $\rho^t = (\rho^t_\text{agents}, \rho^t_\text{items})$ where $\rho^t_\text{agents} \subseteq [n]$ and $\rho^t_\text{items} \subseteq [m]$ are the set of agents yet to be visited and items still available, respectively.
In each round $t$, (1) the mechanism picks an agent $i^t\in \rho^{t-1}_{\text{agents}}$ and posts a price $p_j^t$ for each available item $j\in \rho^{t-1}_\text{items}$; (2) agent $i^t$ selects a bundle $x^t$ from the set of available items and is charged payment $\sum_{j\in x^t}p^{t}_j$; (3) the remaining items, remaining agents, temporary allocation, and temporary payment profile are all updated accordingly. Here, it  is convenient to initialize with  $\rho^0_\text{agents}=[n], \rho^0_\text{items}=[m], \alloc^t=(\emptyset,\ldots,\emptyset)$ and $\p^0 = (0,\ldots,0)$.

\paragraph{Learning Framework.} \label{sec:pomdp}

The sequential nature of  $\MM$s, as well as the private nature of agents' valuations,
makes it useful to formulate this problem of automated mechanism design as a
{\em partially observable Markov decision process} (POMDP).
A POMDP \citep{kaelbling1998planning} is an MDP (given by a state space $\mathcal S$, an action space $\mathcal A$, a Markovian state-action-state transition probability function $\mathbb P(s';s,a)$, and a reward function $r(s,a)$), together with a possibly stochastic mapping from each action and resulting state to observations  $o$ given by $\mathbb P(o;s',a)$.

For $\MM$s, the state corresponds to the items still unallocated, agents not yet visited, a partial allocation, and valuation functions of agents.
An action  determines which agent to go to next and what prices to set. This leads to a new state and observation, namely the item(s) picked by the agent.  In this way,
the state transition is governed by  agent strategies, i.e., the
dominant-strategy equilibrium of $\MM$s. A policy  defines  the rules of the mechanism. An optimal policy for a suitably defined reward function corresponds to an optimal mechanism.
Solving POMDPs requires reasoning about the {\em belief state}, i.e., the belief about the distribution on states given a history of observations. A typical approach  is to find a {\em sufficient statistic} for the belief state, with policies defined as mappings from this statistic to  actions.

\section{Characterization Results}\label{sec:mechanisms}

In $\MM$s, the outcomes from previous rounds can be used to decide which agent to visit and what prices to set in the current round. This allows  prices to be personalized and adaptive, and it also allows the order in which agents are visited to be adaptive.
We next introduce some special cases.
\begin{definition}[Anonymous static price (ASP) mechanisms]
	Prices are set at the beginning (in a potentially random way) and are the same across rounds and for every agent.
\end{definition}
An example of a mechanism in the ASP class is the static pricing mechanism in \citet{FeldmanGL15}.
\begin{definition}[Personalized static price (PSP) mechanisms]
	Prices are set at the beginning (in a potentially random way) and are the same across rounds, but each agent might face different prices.
\end{definition}
Beyond prices, we are also interested in the order in which agents are selected by the mechanism:
\begin{definition}[Static order (SO) mechanisms]
	The order is set at the beginning (in a potentially random way) and does not change across rounds.
\end{definition}

We illustrate the relationship between the various mechanism classes in Figure~\ref{taxonomy}.

The ASP class is a subset of the PSP class, which is a subset of $\MM$.\footnote{As with PSP mechanisms, there exist ASP mechanisms that can  take useful advantage of adaptive order (while holding prices fixed); see Proposition \ref{dynamic_order_corr}.}
Serial dictatorship (SD) mechanisms are a subset of ASP (all payments are set to zero) and may have adaptive or static order. The  {\em random serial dictatorship mechanism} (RSD)~\citep{abdulkadirouglu1998random} lies in the intersection of SD and static order (SO).

\newcommand{\flowFigFontSize}{\footnotesize}
\newcommand{\flowFigCrlSize}{30mm}
\newcommand{\flowFigLabelDist}{-.2cm}
\newcommand{\circlesize}{10mm}
\begin{figure}[t!]
	\centering
	\begin{tikzpicture}
	[ font = \flowFigFontSize, line width=1pt,
	my circle0/.style={minimum width=\flowFigCrlSize*.85, circle, draw},
	my circle1/.style={minimum width=\flowFigCrlSize*.75, circle, draw},
	my circle2/.style={minimum width=\flowFigCrlSize*.55, circle, draw},
	my circle/.style={minimum width=\flowFigCrlSize * 1.5, circle, draw},
	my circle3/.style={minimum width=\flowFigCrlSize * 2, circle, draw},
	my label/.style={left = 0, anchor=mid}
	]
	\node [label = PSP] (1) [my circle, left = -.5 * \flowFigCrlSize] {};
	\node [label = SO] (2) [my circle, right = -.5 * \flowFigCrlSize] {};
	\node [label={[label distance=\flowFigLabelDist]125:ASP}] (3) [my circle0, left = -.05*\flowFigCrlSize] {};
	\node [label={[label distance=\flowFigLabelDist]125:SD}] (5) [my circle2, above= -\flowFigCrlSize*.85 of 3] {};
	\node [label=SPM] (6) [my circle3] {};
	\node[label={[label distance=\flowFigLabelDist]90:RSD}] (7) at (-\flowFigCrlSize*.375,-\flowFigCrlSize*.425+\flowFigCrlSize*.275) [] {\textbullet};
	\end{tikzpicture}
	\caption{The Sequential Price Mechanism ($\MM$) Taxonomy. \label{taxonomy}}
\end{figure}
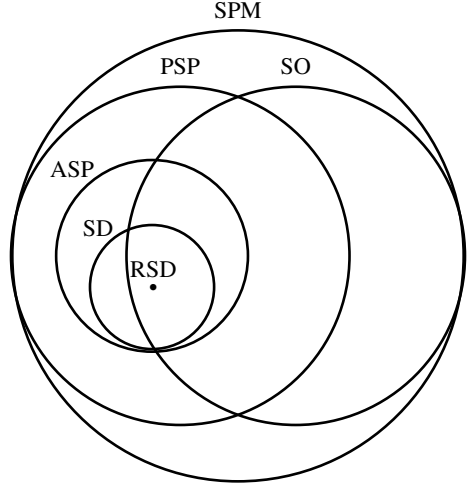

\subsection{The Need for Personalized Prices and Adaptiveness}
\label{sec:th-static}
In this section, we show that  personalized prices  and adaptiveness are necessary for optimizing welfare,  even in surprisingly simple settings. This further motivates formulating the design problem as a POMDP and using RL methods to solve it. We return to the examples embodied in the proofs of these propositions in our experimental work.

Define a {\em welfare-optimal $\MM$} to be a mechanism that optimizes expected social welfare over the class of $\MM$s.
\begin{proposition}
	\label{personalized_price}
	There exists a setting with one item and two IID agents where the welfare-optimal $\MM$ mechanism must use personalized prices.
\end{proposition}
\begin{proof}
	Consider a setting with one item and two IID agents where each has a valuation distributed uniformly on the set $\{1,3\}$. Note that it is WLOG to only consider prices of $0$ and $2$. One optimal mechanism first offers the item to agent $1$ at price $p^1=2$. Then, if the item remains available, the mechanism offers the item to agent $2$ at price $p^2=0$. No single price $p$ can achieve OPT. If $p=0$, the first agent visited might acquire the item when they have value $1$ and the other agent has value $3$. If $p=2$, the item will go unallocated if both agents have value $1$.
\end{proof}

Note that  an adaptive order would not eliminate the need for personalized prices in the example used in the proof of Proposition~\ref{personalized_price}. Interestingly,  we  need $\MM$s with adaptive prices even with IID agents and identical items.
\begin{restatable}{proposition}{dynamicprices}
	\label{dynamic_price}
	There exists a unit-demand setting with two identical items and three IID agents where the welfare-optimal $\MM$  must use adaptive prices.
\end{restatable}
We provide a proof sketch, and defer the proof to Appendix~\ref{sec:missingproofsmechanisms}. The need for adaptive prices comes from the need to be responsive to the remaining supply of items after the decision of the first agent: (i) if this agent buys, then with one item and two agents left,  the optimal  price should be high enough to allocate the item to a high-value agent, alternatively (ii) if this agent does not buy,  subsequent prices should be low to ensure both remaining items are allocated.

The following proposition shows that  an adaptive order may be necessary, even when the optimal prices are anonymous and static.
\begin{restatable}{proposition}{dynamicordercorr}\label{dynamic_order_corr}
	There exists a unit-demand setting with two identical items and six agents with correlated valuations where the welfare-optimal $\MM$ must use an adaptive order (but anonymous static prices suffice).
\end{restatable}
We defer the proof to Appendix~\ref{sec:missingproofsmechanisms}. The intuition is that the agents' valuations are dependent, and knowing one particular agent's value gives important insight into the conditional distributions of the other agents' values. This ``bellweather" agent's value can be inferred from their decision to buy or not, and this additional inference is necessary for ordering the remaining agents optimally. Thus the mechanism's order must adapt to this agent's  decision.

Even when items are identical, and agents' value distributions are independent, both adaptive order and adaptive prices may be necessary.
\begin{restatable}{proposition}{dynamicorderprice}\label{dynamic_order_price}
	There exists a unit-demand setting with two identical items and four agents with independently (non-identically) distributed values where the welfare-optimal $\MM$  must use both adaptive order and adaptive prices.
\end{restatable}
We defer the proof to Appendix~\ref{sec:missingproofsmechanisms}. The intuition is that one agent has both a higher ``ceiling" and higher ``floor" of value compared to some of the other agents. It is optimal for the mechanism to visit other agents in order to determine the optimal prices to offer this particular agent, and this information-gathering process may take either one or two rounds.
We present additional, fine-grained results regarding the need for adaptive ordering of agents for $\MM$s in in Appendix~\ref{sec:personalizedchar}.

\section{Learning Optimal $\boldsymbol{\MM}$s}
\label{sec:mech-pomdp}

In this section, we cast the problem of designing an optimal $\MM$ as a POMDP problem.
Our discussion mainly relates to welfare maximization, but we will also comment on how our results extend to revenue maximization and max-min fairness.

We define the POMDP as follows:
\begin{itemize}
	\item A state $s^t = (\vals,\alloc^{t-1},\rho^{t-1})$ is a tuple consisting of the agent valuations $\vals$, the current partial allocation $\alloc^{t-1}$ and the residual setting $\rho^{t-1}$ consisting of agents not yet visited and items not yet allocated.
	\item An action $a^t = (i^t, p^t)$ defines the next selected agent $i^t$ and the posted prices $p^t$. \item For the state transition, the selected agent  chooses an item or bundle of items $x^t$, leading to a new state $s^{t+1}$, where the bundle $x^t$ is added to  partial allocation $\alloc^{t-1}$ to form a new partial allocation $\alloc^{t}$, and the items and agent are removed from the residual setting $\rho^{t-1}$ to form $\rho^{t}$.
	\item The observation $o^{t+1} = x^t$ consists of the item or set of items $x^t$ chosen by the agent selected at round $t$.
	\item We only provide rewards in terminal states, when the mechanism outcome $\alloc$, $\p$ is available. These terminal rewards are given $\mathsf{g}(\alloc,\p;\vals)$; that is, the objective function we want to maximize.\footnote{We note that, depending on the objective at hand, one can design intermediate reward schemes (e.g., under welfare maximization, value of agent $i^t$ for bundle $x^t$) that may improve learning performance. We choose to only provide final rewards in order to support objectives that can be calculated only given the final outcome, such as max-min fairness.}
\end{itemize}

Next, we study the information that suffices to determine an optimal action after any history of observations. We show the analysis is essentially tight for the case of unit-demand valuations and the social welfare objective. We defer the proofs to Appendix~\ref{sec:missing-proofs-mech-pomdp}.
\begin{restatable}{proposition}{agentsitemssuff}\label{prop:agents_items_suff}
	For agents with independently (non-identically) distributed valuations, with the objective of maximizing welfare or revenue, maintaining remaining agents and items suffices to determine an optimal policy.
\end{restatable}
Interestingly, the statement in Proposition~\ref{prop:agents_items_suff} is no longer true when dealing with a more allocation-sensitive objective such as max-min fairness.\footnote{Consider an instance where some agents have already arrived and been allocated, and the policy can either choose action $a$ or $b$. Action $a$ leads to a max-min value of yet to arrive agents of $5$ with probability $1/2$, and $1$ with probability $1/2$. Action $b$ leads to a max-min value of yet to arrive agents of $10$ with probability $1/2$, and $0$ with probability $1/2$. If the max-min value of the partial allocation is $2$, then the optimal action to take is action $a$. However, if the max-min value of the partial allocation is $10$, then the optimal action is $b$. In particular, inference about the values of agents already allocated is necessary to support optimal actions, and the simple remaining agents/items statistic is not sufficient.} The next theorem reasons about history information for all distributions and objectives.
\begin{restatable}{theorem}{thmone}\label{thm:1}
	With correlated valuations, the allocation matrix along with the agents who have not yet received an offer is sufficient to determine an optimal policy, whatever the design objective.
	Moreover, there exists a unit-demand setting with correlated valuations where optimal policies must use information of size $\Omega\left(\min\{n,m\}\log\left(\max\{n,m\}\right)\right)$.
\end{restatable}
For sufficiency, the allocation matrix and remaining agents always suffices to recover the entire history of observations of any (deterministic) policy. The result follows, since there always exists  deterministic, optimal policies for POMDPs given the entire history of observations (this follows by the Markov property~\citep{bellman1957markovian}).
Theorem~\ref{thm:1} also establishes that carrying the current allocation and remaining agents is necessary from a space complexity viewpoint, since the current allocation and remaining agents can be encoded in  $O\left(\min\{n,m\}\log\left(\max\{n,m\}\right)\right)$ space. 
Another direct corollary is that knowledge of the remaining agents and  items (linear space), and not decisions of previous agents, is not in general enough information to support optimal policies. The problem that arises with correlated valuations comes from the need for inference about the valuations of remaining agents.

As the next proposition shows, policies that can only access remaining agents correspond to a special case of $\MM$s.

\begin{restatable}{proposition}{personalizedpricesstate}\label{prop:personalized-prices-state}
	The subclass of SPMs with static, possibly personalized prices, and a static order, corresponds to policies that only have access to the set of remaining agents.
\end{restatable}

\paragraph{Linear Policies are Insufficient.}

Given access to the allocation matrix and  remaining agents, it is also interesting to understand the class of policies that are necessary to support the welfare-optimal mechanisms. Given input parameters $x$, linear policies map the input to the $\ell$th output using a linear transformation $x\cdot\theta_\ell^\mathsf{T}$, where  $\theta = \{\theta_\ell\}_\ell$ are parameters of the policy. For the purpose of our learning framework, $x$ is a flattened binary  allocation matrix and a binary vector of the remaining agents. We output $n+m$ output variables representing the scores of  agents (implying an order), and the prices of  items.
We are able to show that linear policies are insufficient.

\begin{restatable}{proposition}{nonlinearpolicy}\label{prop:non-linear-policy}
	There exists a setting where the welfare-optimal $\MM$  cannot be implemented via a  policy that is linear in the allocation matrix and  remaining agents.
\end{restatable}

This provides support for non-linear methods for the $\MM$ design problem, motivating the use of neural networks.

\section{Experimental Results}
In this section, we test the ability of standard RL algorithms to learn optimal $\MM$s  across a wide range of settings.

\paragraph{RL Algorithm.}

Motivated by its good performance across different domains, we report our results for the \textit{proximal policy optimization} (PPO) algorithm~\citep{schulman2017proximal}, a policy gradient algorithm where the learning objective is modified to prevent large gradient steps, and as implemented in OpenAI Stable Baselines.\footnote{We use the OpenAI Stable Baselines version v2.10.0 (https://github.com/hill-a/stable-baselines).} Similarly to \citet{wu2017scalable, mnih2016asynchronous}, we run each experiment using 6 seeds and use the 3 seeds with highest average performance to plot the learning curves in figures
~\ref{fig:exp_simpleCorrelated}~-~\ref{fig:more-settings}.  At periodic intervals during training, we evaluate the objective of the current policy using a fresh set of samples. It is these evaluation curves that are shown in our experiment figures. “Performance” means average objective value of the three selected seeds--objective value is welfare, revenue, or max-min fairness, depending on the setting.
The shaded regions show 95\% confidence intervals based on the average performances of the 3 selected seeds. This is done to plot the benchmarks as well.

We encode the policy via a standard 2-layer \textit{multilayer perceptron} (MLP) \citep{bourlard1989speech} network. The policy takes as input a statistic of the history of observations (different statistics used are described below), and outputs $n + m$ output variables, used to determine the considered agent and the  prices in a given round.
The first $n$ outputs give agents' weights,
and agent $i^t$ is selected as the highest-weight agent among the remaining agents using a $\argmax$ over the weights. The other $m$ weights give the prices agent $i^t$ is faced.
The state transition function models agents that follow their dominant strategy, and pick a
utility-maximizing bundle given offered prices.

At the end of an episode, we calculate the reward. For social welfare, this reflects the allocation and agent valuations; other objectives can be captured, e.g., for revenue the reward is the total payment collected, and for max-min fairness, the reward is the minimum value across agents.
We also employ variance-reduction techniques, as is common in the RL literature~\cite[e.g.]{greensmith2004variance}.\footnote{For welfare and revenue, we  subtract the optimal welfare from the achieved welfare at each episode. As the optimal welfare does not depend on the policy, a policy maximizing this modified reward also maximizes the original objective.}

In order to study trade-offs between simplicity and robustness of learned policies, we vary the statistic of the history of observations that we make available to the policy:
\begin{enumerate}
	\item \textit{Items/agents left}, encoding which items are still available and which agents are still to be considered. As discussed above, this statistic supports optimal policies when agents have independently distributed valuations for welfare and revenue maximization.
	\item \textit{Allocation matrix} that, in addition to items/agents left, encodes the temporary allocation $\alloc^t$ at each round $t$. As discussed above, this statistic supports optimal policies even when agents' valuations are correlated and for all objectives.
	\item \textit{Price-allocation matrix}, which, in addition to items/agents left and temporary allocation, stores an $n\times m$ real-valued matrix with the prices the agents have faced so far. This is a sufficient statistic for our POMDPs as it captures the entire history of observations.
\end{enumerate}

\paragraph{Baselines.}
We consider the following three  baselines:
\begin{enumerate}
	\item \textit{Random serial dictatorship}, where the agents' order is determined randomly,  and prices are set to zero.
	\item \textit{Anonymous static prices}, where we constrain  policies to those that correspond to  ASP mechanisms (this is achieved by hiding all history from the policy, which forces the order and prices not to depend on past observation or the identity of the next agent).
	\item \textit{Personalized static prices}, where we constrain  policies to  the family of PSP mechanisms (this is achieved by only providing the policy with information about the remaining agents; see Proposition~\ref{prop:personalized-prices-state}).
\end{enumerate}

\begin{figure*}[h!]
\centering
\includegraphics[width=.95\linewidth]{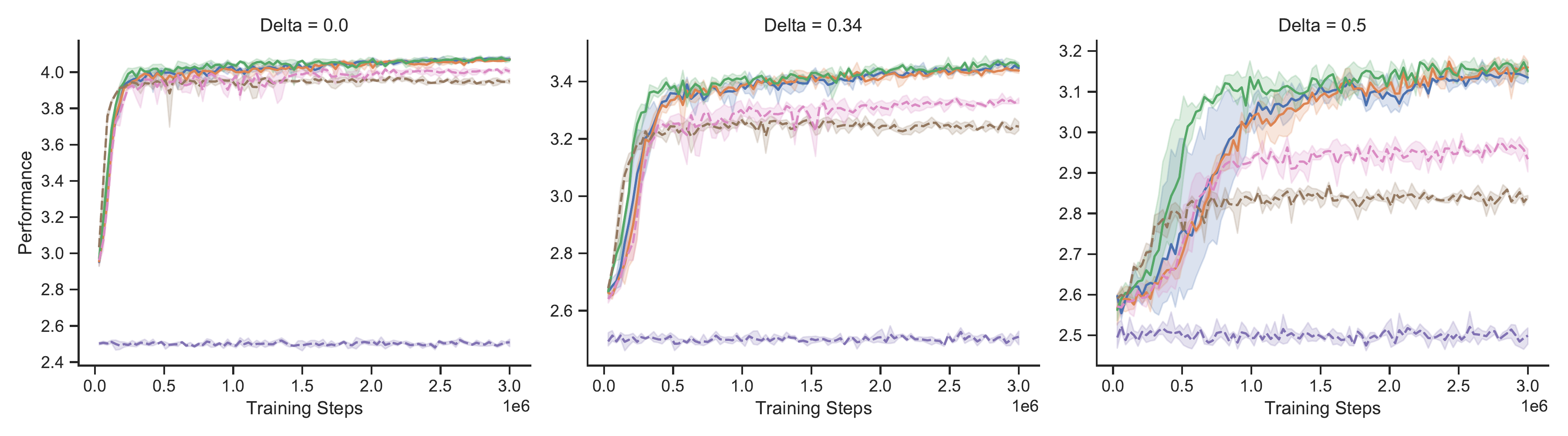}
\caption{Corr.~Value, welfare objective. 20 agents, 5 identical items, varying corr.~parameter, $\delta$.  See Figure~\ref{fig:basic-experiments} for legend. \label{fig:exp_simpleCorrelated}}
\includegraphics[width=.95\linewidth]{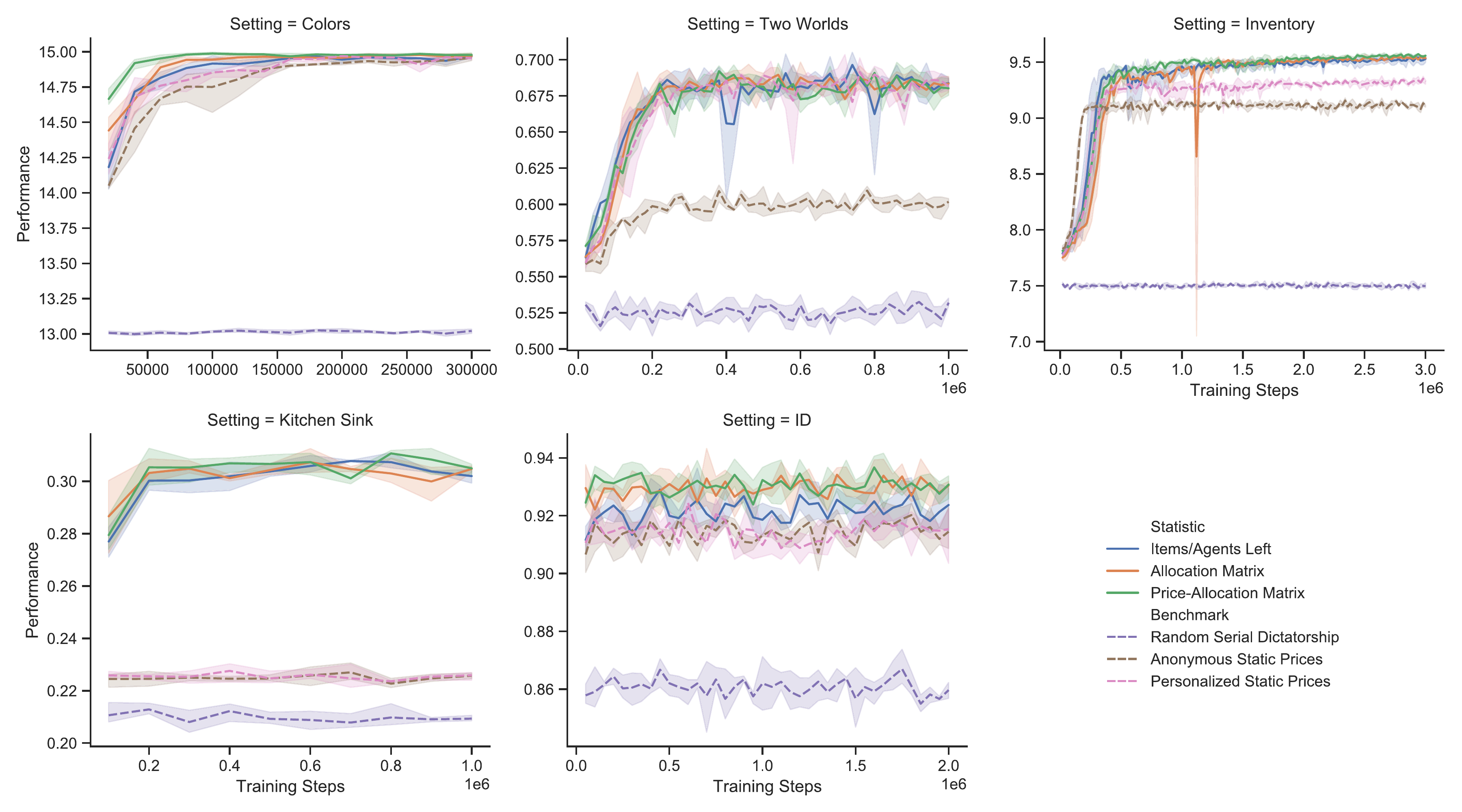}
\caption{Theory-driven, welfare objective. (a) Colors. (b) Two Worlds. (c) Adaptive pricing. (d) Adaptive order and pricing. (e) Allocation information. \label{fig:basic-experiments}}
\includegraphics[width=.95\linewidth]{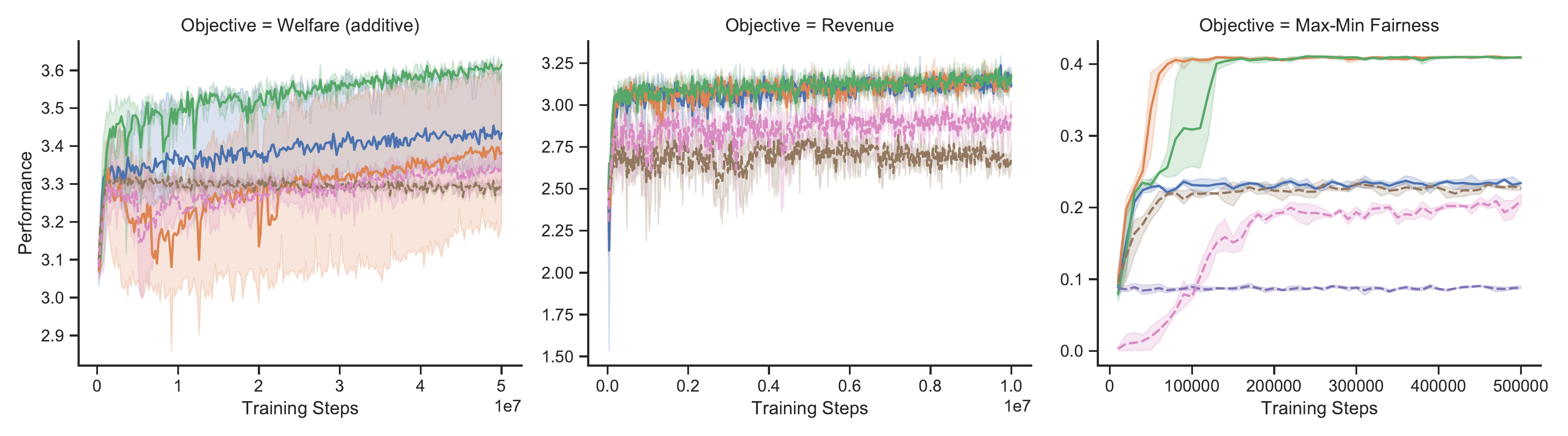}
\caption{Beyond UD and WM. (a) Additive across types, welfare objective (corr.~setting, 10 agents, 2 \& 4 identical items, $\delta = 0.5$). (b) Revenue objective (corr.~setting, 20 agents, 5 items, $\delta = 0.5$). (c) Max-min fairness. See Figure~\ref{fig:basic-experiments} for legend. \label{fig:more-settings}}
\end{figure*}

\paragraph{Part 1: Correlated Value Experiments (Welfare).}
Recognizing the role of correlations in the power that comes from the adaptivity of $\MM$s, we first test a setting with multiple identical copies of an item, and agents with unit-demand and correlated values. For this, we use parameter $0 \leq \delta \leq 1$ to control the amount of correlation.  We sample  $z\sim U[\frac{1-\delta}{2}, \frac{1+\delta}{2}]$, and draw $v_i$ independently from $\mathrm{unif}(z-\frac{1-\delta}{2}, z+\frac{1-\delta}{2})$. For $\delta = 0$ this gives i.i.d. $v_i$ all drawn uniformly between $0$ and $1$. For $\delta = 1$ this gives all identical $v_i = z$. For intermediary values of $\delta$ we get increasing correlation between the $v_i$'s.

The results are reported in Figure~\ref{fig:exp_simpleCorrelated}. We vary the number of agents, items, and  $\delta$, controlling the level of correlation. We show results for 20 agents and 5 identical items, and $\delta = 0, 0.25, 0.33$, and $0.5$. The POMDP with the price-allocation matrix statistic is able to substantially outperform the best static mechanism as well as RSD. A dynamic approach using an allocation matrix or agents and items left also outperforms a static mechanism, but learns more slowly than an RL policy that is provided with a price history, especially for larger $\delta$. Results for other combinations of agents and items (up to 30 each were tested) yield similar results.\footnote{Experiments with a small number of items, or close to as many items as agents, yield less-interesting results, as these problems are much easier and all approaches achieved near-optimal welfare.}

\paragraph{Part 2: Theory-driven Experiments (Welfare).}

Second, we look to support the theoretical results described above. We consider five different settings, each with unit-demand agents. We defer the full description of the settings to Appendix~\ref{sec:theory-driven-experiments}.
In each of the settings, the optimal $\MM$ mechanism has different features:
\begin{itemize}
	\item  \textit{Colors}: the optimal $\MM$   is an anonymous static pricing mechanism.
	\item \textit{Two worlds}: the optimal $\MM$  is a static mechanism but requires personalized prices.
	\item \textit{Inventory}:  the optimal  $\MM$ makes use of adaptive prices, and this  outperforms the best  static personalized price mechanism, which outperforms  the best static and anonymous price mechanism.
	\item \textit{Kitchen sink}:  both types of adaptiveness are needed by the optimal $\MM$.
	\item \textit{ID}: the statistic of remaining agents and items is not sufficient to support the optimal policy.
\end{itemize}
Figure~\ref{fig:basic-experiments} shows the results for the different setups.
Our experiments show that (a) we are able to learn the optimal $\MM$ mechanism for each of the setups using deep RL algorithms; and (b) we are able to show exactly the  variation in performance suggested by theory, and depending on the type of statistics used as input for the policy:
\begin {itemize}
\item In Figure~\ref{fig:basic-experiments} (a) (Colors) we  get optimal performance already when learning a static anonymous price policy.
\item In Figure~\ref{fig:basic-experiments} (b) (Two worlds) a static personalized price policy performs optimally, but not a static anonymous price policy.
\item Figure~\ref{fig:basic-experiments} (c) (Inventory)  adaptive policies are able to achieve optimal performance, outperforming personalized price mechanisms, which in turn outperform anonymous price mechanisms. \item Figure~\ref{fig:basic-experiments} (d) (Kitchen sink) adaptive policies are able to learn an optimal policy that requires using both adaptive order and adaptive prices.
\item Finally, Figure~~\ref{fig:basic-experiments} (e) (ID)  some setups require more complex information, as policies that leverage allocation information outperform the policy that just access remaining agents and items.
\end{itemize}

\paragraph{Part 3: Beyond Unit Demand, and Beyond Welfare Maximization.}
Third, we present  results for more general setups (we defer their full description to Appendix~\ref{sec:other-settings}):
\begin{itemize}
\item   \textit{Additive-across-types under welfare objective}: there are two item types, and agents have additive valuations on one unit of each type.
\item  \textit{Revenue maximization}: we work in the correlated setting from part one, with $\delta=0.5$, but for a revenue objective. %
\item  \textit{Max-min fairness}: the goal is to maximize the minimum value achieved by an agent in an allocation, and we consider a setting where an adaptive order is required for an optimal reward.
\end{itemize}

See Figure~\ref{fig:more-settings}.
These results show the full generality of the framework, and show the promise in using deep-RL methods for learning  $\MM$s for varying settings. Interestingly, they also show different sensitivities for the statistics used than in the unit-demand, welfare-maximization setting. For the additive-across-types setting, price information has a still greater effect on the learning rate. For the max-min fairness setting,  providing the entire allocation information has a  large effect on the learning process, as the objective  is very sensitive to specific parts of the allocation; this is also consistent with the fact that agents and items left do not provide sufficient information for this objective (see the discussion following Proposition~\ref{prop:agents_items_suff}).

\section{Conclusion}

We have studied the class of $\MM$s, providing characterization results and formulating the optimal design problem as a POMDP problem. Beyond studying the history statistics to support optimal policies,
we have also demonstrated the practical learnability of the class of $\MM$s  in  increasingly complex settings.
This work points toward many interesting open questions for future work. First, it will be interesting to adopt policies with a fixed-size memory, for instance through  LSTM methods~\citep{hochreiter1997long}, allowing the approach to potentially scale-up to very large numbers of agents and items (dispensing with large, sufficient statistics). Second, it will be interesting and challenging to study settings where there is no simple, dominant-strategy equilibrium. This will require  methods to also model agent behavior~\citep{PhelpsMPS02,Byde03,Wellman06,PhelpsMP10,ThompsonL13,BunzLS18,ViqueiraCMG19,zheng20}. Third, it is  interesting to consider settings that  allow for communication between agents and the mechanism, and study the automated design of emergent, one- or two-way communication (c.f.,~\citet{lowe2017multi}).

\section{Acknowledgments}
We  thank Zhe Feng and Nir Rosenfeld for helpful discussions and feedback. The project or effort depicted was or is sponsored, in part, by the Defense Advanced Research Projects Agency under cooperative agreement number HR00111920029, the content of the information does not necessarily reflect the position or the policy of the Government, and no official endorsement should be inferred. Approved for public release; distribution is unlimited. The work of G. Brero was also supported by the SNSF (Swiss National Science Foundation) under Fellowship P2ZHP1191253.

\newpage

%% Bibliography
\bibliographystyle{plainnat}
\bibliography{bibliography}

\appendix
\section{Omitted Proofs of Section~\ref{sec:mechanisms}}\label{sec:missingproofsmechanisms}
\dynamicprices*
\begin{proof}
	Consider a unit-demand setting with three agents, each with value distributed uniformly and independently on the set  $\{1,3\}$. Note that it is WLOG to only consider prices of $0$ and $2$, and to consider both items having the same price. An example of a welfare-optimal mechanism is to fix an arbitrary order and first set price $p^1=2$. If the first agent does not buy, the mechanism sets prices $p^2=p^3=0$ so both items are allocated. If the first agent does buy, the mechanism next sets price $p^2=2$, and then if the item remains available, sets price $p^3=0$. This mechanism achieves the optimal welfare because it always allocates both items, and every agent with value $3$ gets an item (unless all three agents have value $3$).
	
	Any welfare-optimal mechanism must  set $p^3=0$, because if it visits the final agent, it is optimal to allocate the item to that agent unconditionally. The mechanism must set $p^1=2$ or else it risks allocating an item to an agent with value $1$ and leaving an agent with value $3$ without an item. If the first agent visited does not buy, there are two items and two agents left, so $p^2$ must be $0$ or else the mechanism risks leaving an item unallocated. If the first agent does buy, there is one item and two agents, so $p^2$ must be $2$ or else the mechanism risks allocating the last item to the second agent when they have value $1$ and the third agent has value $3$. Thus, the possible price histories under this (strictly) optimal pricing policy are $(p^1=2,p^2=0,p^3=0)$, $(p_1=2,p_2=2)$, and $(p^1=2,p^2=2,p^3=0)$. The first history has one agent face price $2$ and two agents face price $0$, while the third history has the opposite. So, no matter what order is used, some agent must face price $2$ in the first history and price $0$ in the third history. Thus, personalized static prices are insufficient, and adaptive prices are necessary to achieve optimal welfare.
\end{proof}

\dynamicordercorr*
\begin{proof}
	%\dr{To-do, and move to appx}
	Consider a unit-demand setting with six agents and two items. Agent 1's value is drawn uniformly from $\{1,15\}$. If it's $15$, agent 2's value is drawn uniformly from $\{2,12\}$; otherwise it is drawn uniformly from $\{3,8\}$. Agents 3 and 4 are the same as agent 2 but with the distributions switched, so $\{2,12\}$ if $v_1=1$ and $\{3,8\}$ otherwise. Agents 5 and 6 have value $4$ deterministically. 
	
	An example of a welfare-optimal mechanism is to visit agent 1 first and set price $p^1=3.5$. This allocates to agent 1 iff they have value $15$ (if value $1$, any other agent is preferred). Then:
	\begin{itemize}
		\item If agent 1 bought, the mechanism visits agent 2 and sets price $p^2=3.5$. This allocates to agent 2 iff they have value $12$ (if value $2$, any other agent is preferred). If an item remains, visit agents 3, 4, and 5 in turn with price $p^3=p^4=p^5=3.5$, which will allocate to an agent with value $8$ if any exist, else to agent 5 with value $4$. 
		\item If agent 1 did not buy, the mechanism visits agents 3 and 4 in turn with prices $p^2=p^3=3.5$. This allocates to agents 3 and/or 4 iff they have value $12$ (if value $2$, any other agent is preferred). If any items remain, visit agents 2, 5, and 6 in turn with price $p^4=p^5=p^6=3.5$, which will allocate to agent 2 if they have value $8$, else to agents 5 and/or 6 who have value $4$.
	\end{itemize}
	
	Price  $3.5$ works because it lies between the ``high" and ``low" values of agents 1-4 and below the values of agents 5 and 6. Because agents 5 and 6 have higher values ($4$) than the other agents' ``low" values ($1$,$2$,$3$), the mechanism would rather allocate to agents 5 or 6 than to an agent with a low value, so the price does not need to be lower than $3.5$. Because the mechanism visits agents in decreasing order of ``high" value, agents 1-4 will only purchase an item if they have a ``high" value and no other agent has a strictly higher value. So the price need not be higher.
	
	The key feature of this mechanism is that it infers agent 1's value from its decision to buy or not, then uses this information to ensure it visits the $\{2,12\}$ agent(s) before the $\{3,8\}$ agents. Before observing agent 1's decision, the mechanism cannot know which agent(s) are $\{2,12\}$ and which are $\{3,8\}$. 
	
	In particular,  any static order mechanism risks visiting one of agents 2-4 while either a) there are at least two other agents who could have value $12$ or $15$, or b) there is one item and at least one agent who could have value $12$.  In this situation, any price less than $8$ risks allocating an item to this agent with value $8$ and leaving an agent with value $12$ or $15$ without an item. Meanwhile, any price $8$ or greater risks not allocating to this agent when they have value $8$ and that is strictly the greatest among remaining agents. Because no price can ensure optimal welfare in this situation, a static order is insufficient, even with access to adaptive prices.
\end{proof}

\dynamicorderprice*
\begin{proof}
	Consider a unit-demand setting with four agents and two items. Agent 1's value is drawn uniformly from $\{1,15\}$. Agent 2's value is drawn uniformly from $\{3,12\}$. Agents 3 and 4 have values drawn uniformly from $\{2,8\}$. All valuations are independent. 
	
	We will build a welfare-optimal mechanism by backwards induction. Because the mechanism knows each agents' value distribution when it visits them, it is sufficient to choose between a price below the agents' ``low" value, one above their ``high" value, and one between them. As such, it is sufficient to consider only prices of $0$, $5$, and $20$. $0$ is below all agents' ``low" values, $5$ between all ``low" and ``high" values, and $20$ above all ``high" values. Notice also that because the items are identical, it is sufficient to consider a mechanism that sets the same price for each item. 
	
	We observe the following fact:  Every agents' ``low" value is below every other agents' expected value.
	
	This implies it is never optimal to set price $0$ unless there are as many items as agents --- the mechanism would rather allocate to any of the remaining agents unconditionally than to this agent's low value. It also implies it is never optimal to use price $20$ when approaching any agent $a$.  
	If $a$ is the last agent, then of course price 20 will result in no sale, which means lowering the price to 0 strictly improves welfare. If $a$ is not the last agent, then let $b$ be the last agent in the sequence. 	Moving $a$ to the end of the sequence, changing $a$'s price from $20$ to $0$, and changing $b$'s price from 0 to 5 strictly increases the expected welfare (we replace allocations to $b$'s low value with allocations to $a$).
	Therefore, a welfare-optimal pricing uses price $5$ as long as there are more agents than items, and drops to price $0$ once there are as many items as there are agents.
	
	We use  notation  $V(i, S, x)$ to denote the expected value obtained by visited agent $i\in S$ first when there's a set of agents $S$ and $x$ items left. We now perform the backward induction calculation. We simplify the calculations using the fact that agents 3 and 4 are interchangeable.
	The following are immediate from the above discussion:
	$$V(3, \{3\}, 1)= 5,\quad V(2,\{2\},1)=7.5,\quad V(1,\{1\},1)=8,\quad V(3, \{3,4\},1) = 6.5,$$
	$$V(2, \{2,3\},1)=8.5 \ > \ V(3,\{2,3\}, 1),\quad V(1, \{1,3\},1)=10 \ > \ V(3,\{1,3\},1),$$
	$$V(1, \{1,2\},1)=11.25 \ > \ V(2,\{1,2\},1).$$

    \noindent	
	We now calculate the expected value for situations where we have 1 item and 3 agents left.
	$$V(1,\{1,2,3\},1)= 15/2 + V(2, \{2,3\},1)/2 = 11.75 > \max\{V(2,\{1,2,3\},1),V(3,\{1,2,3\},1)\},$$
	$$V(1,\{1,3,4\},1)= 15/2 + V(3, \{3,4\},1)/2 = 10.75 > V(3,\{1,3,4\},1),$$
	$$V(2,\{2,3,4\},1)= 12/2 + V(3, \{3,4\},1)/2 = 9.25 > V(3,\{2,3,4\},1).$$
    
    \noindent	
	And now we calculate the expected value for situations where we have 2 items and 3 agents left.
	$$V(1,\{1,2,3\},2) = (15+V(2,\{2,3\},1))/2 + (5+7.5)/2 = 18 > \max\{V(2,\{1,2,3\},2),V(3,\{1,2,3\},2)\},$$
	$$V(1,\{1,3,4\},2) = (15+V(3,\{3,4\},1))/2 + (5+5)/2 = 15.75 > V(3,\{1,3,4\},2),$$
	$$V(3,\{2,3,4\},2) = (8+V(2,\{2,3\},1)/2 + (7.5+5)/2 = 14.5 > V(2,\{2,3,4\},2).$$
	
    \noindent	
	Finally, we calculate the expected welfare of the optimal policy:
	$$ V(1, \{1,2,3,4\},2) = (15 + V(2,\{2,3,4\},1))/2 + V(3, \{2,3,4\},2)/2 $$ 
	$$ = 19.375 > \max\{V(2, \{1,2,3,4\},2),V(3, \{1,2,3,4\},2) \}.$$
	From the above calculation, it shows it is always optimal to approach agent 1 first with a price of 5. If an item is sold, then the policy approaches agent 2 with a price 5, otherwise, it approaches agent 3 (WLOG agent 3 and not 4) with a price 5; therefore, the order is adaptive. In the case an item is not sold to agent 1 and not sold to agent 3, then it is optimal to price the item for agent 2 at 0 to extract full surplus. In case an item is not sold to agent 1 and is sold to agent 3, it is then optimal to price the item for agent 2 at 5. Therefore, prices are adaptive.
\end{proof}

\section{Omitted Proofs from Section~\ref{sec:mech-pomdp}}\label{sec:missing-proofs-mech-pomdp}

\agentsitemssuff*
\begin{proof}
	When maximizing welfare (revenue), at a current state of the MDP, it is optimal to choose the action that maximizes the expected welfare (revenue) obtained by selling the remaining items to the yet-to-arrive agents. This is because the objectives are linear additive. Then, given a set of items and agents with independently drawn valuations, an agents' probability of purchasing an item at a given price is independent of the history. Because of this, the only relevant information, for any history of observations, is the set of agents and items left. 
\end{proof}

\personalizedpricesstate*
\begin{proof}
	Without access to the current, partial allocation, neither the agent order nor prices can be adaptive because they only use information that was known prior to the first policy action. Prices can be personalized, however, because the same agent can always be visited at a particular position in the order, and thus receive different prices than other agents.
\end{proof}

\nonlinearpolicy*
\begin{proof}
	Consider a setting with four agents,  Alice, Bob, Carl and Dan, and three items, $1,2,3$. Alice and Bob always have value $0$ for item $3$, and their value for items $1,2$ is distributed uniformly over the following six options: 	
	\begin{eqnarray*}
		& &v_\mathrm{A}(1)=10,\ v_\mathrm{A}(2)=10,\ v_\mathrm{B}(1)=10,\ v_\mathrm{B}(2)=10,\\
		& &v_\mathrm{A}(1)=0,\ v_\mathrm{A}(2)=0 ,\ v_\mathrm{B}(1)=0,\ v_\mathrm{B}(2)=0,\\	
		& &v_\mathrm{A}(1)=10,\ v_\mathrm{A}(2)=0 ,\ v_\mathrm{B}(1)=0,\ v_\mathrm{B}(2)=0,\\
		& &v_\mathrm{A}(1)=0,\ v_\mathrm{A}(2)=10 ,\ v_\mathrm{B}(1)=0,\ v_\mathrm{B}(2)=0,\\
		& &v_\mathrm{A}(1)=0,\ v_\mathrm{A}(2)=0 ,\ v_\mathrm{B}(1)=10,\ v_\mathrm{B}(2)=0,\\
		& &v_\mathrm{A}(1)=0,\ v_\mathrm{A}(2)=0 ,\ v_\mathrm{B}(1)=0,\ v_\mathrm{B}(2)=10.		
	\end{eqnarray*}

    \noindent	
	The value of Carl and Dan for items $1$ and $2$ is always $0$. 
	
	If there exists an allocation for Alice and Bob with welfare $20$, or there does not exist an allocation for Alice and Bob with value larger than $0$, then Carl's value for item $3$ is $10$ w.p. $1/10$ and $0$ w.p. $9/10$, and Dan's value is a constant $5$. If the best allocation for Alice and Bob yields welfare $10$, Carl and Dan reverse roles, and Dan's value for $3$ is $10$ w.p. $1/10$ and $0$ w.p. $9/10$, and Carl's value is a constant $5$.
	
	%It is sufficient to use price $3$ for all items in the welfare-optimal policy. \dcp{is the former claim supposed to be clear, or shown in the following. explain} 
	In order to extract optimal welfare from Carl and Dan, an optimal policy should order Carl before Dan if Dan has a constant value $5$ for this item, and order Dan before Carl if Carl has a constant value. To get the correct order in which Carl and Dan should go, an optimal policy should order Alice and Bob first.
	
	Assume by symmetry that the optimal policy orders Alice before Bob, and assume  WLOG that if the optimal attainable welfare by Alice and Bob is 20, then Alice is allocated item 1 and Bob is allocated item 2 (a policy might force a certain allocation for this case, but both items must be allocated). Consider a linear policy $\theta$. Consider the following four allocation profiles of Alice and Bob: (i) Alice is allocated item 1, Bob is allocated item 2; (ii) Alice is allocated item 1, Bob is unallocated; (iii) Alice is unallocated, Bob is allocated item 2; and (iv) Alice and Bob are unallocated. In all four scenarios, all input variables of the policy are identical, but two variables take on different values, i.e., $x_A^1$  ($x_B^2$) takes on value 1 if  Alice (Bob) is allocated item 1 (2) and 0 otherwise. 
	
	Recall that we consider policies that calculate scores for each agent given the current input variables, and approach the highest-score agent of the agents yet to arrive.
	Let $\theta_C^1$ ($\theta_D^1$) denote the weight multiplied with variable $x_A^1$ to determine the score of Carl (Dan), and let $\theta_C^2$ ($\theta_D^2$) be the weight multiplied with variable $x_B^2$ to determine the score of Carl (Dan). Let $\theta_C^{-12}$ and $\theta_D^{-12}$ be the rest of the weights used to determine the scores of Carl and Dan, and let $x^{-12}$ bet the rest of input variables besides  $x_A^1$ and $x_B^2$. 
	
	Since in (i) Carl has a higher score than Dan, we have:
	\begin{align}
	x_A^1\cdot \theta_C^1 + x_B^2\cdot \theta_C^2 +  x^{-12}\cdot \theta_C^{-12} & > x_A^1\cdot \theta_D^1 + x_B^2\cdot \theta_D^2 +  x^{-12}\cdot \theta_D^{-12}  \notag
	\\
	\Rightarrow  \quad \theta_C^1 +  \theta_C^2 +  x^{-12}\cdot \theta_C^{-12} & >\theta_D^1 +  \theta_D^2 +  x^{-12}\cdot \theta_D^{-12}. \label{eq:i}
	\end{align}
	
	    \noindent	
	Similarly, from (ii), (iii) and (iv) we have
	\begin{align}
	\theta_C^1  +  x^{-12}\cdot \theta_C^{-12} &< \theta_D^1 +    x^{-12}\cdot \theta_D^{-12}\label{eq:ii}
	\\ 
	\theta_C^2 +  x^{-12}\cdot \theta_C^{-12} &<    \theta_D^2 +  x^{-12}\cdot \theta_D^{-12}. \label{eq:iii} \\ 
	x^{-12}\cdot \theta_C^{-12} & >     x^{-12}\cdot \theta_D^{-12} \label{eq:iv}
	\end{align}	
    \noindent	
	Subtracting Eq.~\eqref{eq:iv} from Eq.~\eqref{eq:ii} gives:
	\begin{eqnarray}
	\theta_C^1 < \theta_D^1.\label{eq:v}
	\end{eqnarray}
	
	    \noindent	
	Adding Eq.~\eqref{eq:v} to Eq.~\eqref{eq:iii} gives:
	\begin{eqnarray*} \theta_C^1 +  \theta_C^2 +  x^{-12}\cdot \theta_C^{-12} <  \theta_D^1 +  \theta_D^2 +  x^{-12}\cdot \theta_D^{-12},
	\end{eqnarray*}
	contradicting Eq~\eqref{eq:i}.
\end{proof}

\subsection{Proof of Theorem~\ref{thm:1}}
In this section, we prove the following theorem.
\thmone*

We prove this theorem through Lemma~\ref{lem:alloc_suff}, which shows the sufficiency part, and Lemma~\ref{lem:lb-space}, which gives a lower bound on the statistic's space complexity.
\begin{lemma}\label{lem:alloc_suff}
	For any value distribution, the  allocation matrix along with the agents who have not yet received an offer is sufficient to determine the optimal policy, whatever the design objective.
\end{lemma}
\begin{proof}
	The observable history that a policy generates is the agent  approached in each round, the prices offered, and the items  purchased  (if any). 
	We first notice that given the entire observable history, it is without loss to assume the optimal policy is deterministic (this follows, for instance, from \citet{bellman1957markovian}). 
	
	We show that for any fixed, deterministic policy, that the allocation information and the list of visited agents is sufficient to fully recover the entire observable history. Therefore, this information is sufficient in order to take the optimal action of the optimal policy. The proof is by induction. Observations are agent visited, prices offered, and item(s) selected by the agent. The base case is the empty history. For the  inductive case, consider that  knowledge of the sequence of observations so far and the new allocation matrix and list of remaining agents reveals the (i) agent just visited, (ii) price offered since the policy is deterministic and this follows from the sequence of observations so far, and (iii) the item(s) selected by the agent.
\end{proof}

We can use this lemma to upper-bound the space complexity of the statistic necessary to support an optimal policy:
\begin{corollary}\label{clm:ub-nlogm}
	For unit-demand bidders, there exists an optimal policy that uses a statistic of the observation history of space complexity $O(n\log m)$.
\end{corollary}
\begin{proof}
	Follows from the fact that the number of allocations of unit-demand bidders in a market with $n$ agents and $m$ items is $O(m^n)$, which takes $O(n\log m)$ bits to represent.
\end{proof}

\begin{corollary}\label{clm:ub-mlogn}
	For any valuations of agents,  there exists an optimal policy that uses a statistic of space complexity $O(m\log n)$.
\end{corollary}
\begin{proof}
	In every allocation, each item has $n+1$ possible options to be allocated (the +1 is for the option it is not allocated). Therefore, the number of allocations is $O(n^m)$, which takes $O(m\log n)$ bits to represent. 
\end{proof}

\begin{lemma}
	There exists a unit-demand setting with correlated valuations where the optimal policy must use a statistic of size $\Omega\left(\min\{n,m\}\log\left(\max\{n,m\}\right)\right)$.\label{lem:lb-space}
\end{lemma}
\begin{proof}
	
	Let $N=\{i_1,\ldots, i_n\}$ and $M=\{1,\ldots, m\}$ be the set of agents and items. In addition, there are two special agents $a$ and $b$, and a special item $x$, who's value will depend on the matching of agents in $N$ to items in $M$.

	We first prove the lemma for the case that $m>2n$.
	The valuations of agents in $N$ are realized as follows:
	
	\begin{itemize}
		\item Set $L=\emptyset$.
		\item For $\ell =1,\ldots, n$:
		\begin{enumerate}
			\item Choose $j_\ell$ uniformly at random from $M\setminus L$.
			\item $L:= L\cup \{j_\ell\}$.
			\item $v_{i_\ell}(j)=\begin{cases}
			1\quad j_\ell\\
			0\quad \mbox{otherwise}
			\end{cases}$.	
		\end{enumerate}
	\end{itemize}
	
	The agents $a$ and $b$  have an i.i.d.~valuation that depends on the matching between $N$ and $M$, which is  a scaled version of the example in Proposition~\ref{personalized_price}, with item $x$ serving as the item sold in the argument used in the proof of that proposition. Therefore, depending on the matching between $N$ and $M$, the price of the first agent in $\{a,b\}$ should be different.
	The number of possible matchings between $N$ and $M$ is $$\binom{m}{n}n!= \frac{m!}{(m-n)!}= m\cdot(m-1)\cdot\ldots\cdot(m-n+1)>(m/2)^n,$$ where we use the fact that $n<m/2$ in the last inequality. To represent this many potential prices, we need a state of size $\Omega(n\log m)$.
	
	For the case of $m=O(n)$, the construction is very similar to the case that $m>2n$. The setting first draws one of $\Omega(m^n)$ possible matchings between $n$ and $m$, and according to the matching drawn, the algorithm should use a different price for $a$, $b$ and item $x$. To represent $\Omega(m^n)$ possible matchings, we need $\Omega(n\log m)$ space.	
\end{proof}
    \noindent	
This completes the proof of Theorem~\ref{thm:1}. A corollary of Corollaries~\ref{clm:ub-nlogm}, \ref{clm:ub-mlogn}, and Lemma~\ref{lem:lb-space} is the following:

\begin{corollary}
	The space complexity of the optimal policy's statistic for unit-demand bidders is\newline $\Theta\left(\min\{n,m\}\log\left(\max\{n,m\}\right)\right)$.
\end{corollary}

\section{Experimental Results}
In this section, we provide an extended description of the settings we tested in Part 2 and Part 3 of the experimental section of our paper. When running our experiments, we normalize the valuations such that the highest possible value is 1. This is done by dividing each value of the settings listed below by the highest possible value in that setting.

\subsection{Part 2: Theory-driven Experiments (Welfare)}\label{sec:theory-driven-experiments}

\paragraph{Colors.}
In this environment, there are $30$ unit-demand agents, $10$ red, $10$ yellow,  and $10$ blue, and
there are also $20$ items,  these items including $10$ red  items and $10$ yellow items. Red agents have value $1$ for each red item and $0$ for each yellow item. Yellow agents have value $1$ for each yellow item and $0$ for each red item.  
The valuation of each blue agent is defined as follows:
\begin{itemize}
	\item Draw $x\sim U[0,1].$
	\item With probability $x$, the agent has value $2$ for each red item and $0$ for each yellow one. Otherwise, the agent has value $2$ the yellow items and $0$ the red ones.
\end{itemize}

In a welfare-optimal allocation, each blue agent receives their preferred item type and the remaining red items go to red agents and yellow items to yellow agents. The optimal social welfare is $30$. In this environment, there exists an optimal static $\MM$  that prices $p\in(0,1)$ each item and lets blue agents go first. %This way, only ``high" yellow and red agents get allocated.

\paragraph{Two Worlds.}
In this environment, there are $10$ unit-demand agents and $1$ item. Valuations are realized as follows:
\begin{itemize}
	\item With probability $1/2$, we are ``high regime:'' each value is drawn uniformly at random from the set $\{0.6,1\}$.
	\item Otherwise, we are in ``low regime,'' and each value is drawn uniformly at random from the set $\{0.1,0.4\}$.
\end{itemize}
For a welfare-optimal allocation, it is enough to use a static $\MM$  with personalized prices, e.g., $p_i=0.9$ for $i\in\{1,..,5\}$ and $p_i = 0.2$ for $i\in\{6,..,10\}$.  

\paragraph{Inventory.}
In this environment, we have $20$ unit-demand agents and $10$ identical items. Each agent's value is drawn uniformly at random from the set $\{1/2,1\}$. 
Here, a welfare-optimal $\MM$  needs to adapt prices depending on agents' purchases: In the first few rounds, each item is priced $p \in (0.5,1)$. Then, when the number of remaining items is equal to the number of remaining agents, the price is lowered to below $0.5$. Note that this welfare-optimal outcome cannot be achieved by a personalized, static price mechanism, as the number of the round at which the price should drop depends on the agents' demand and cannot be determined ex ante.

\paragraph{Kitchen Sink.}
This environment consists of 3 unit-demand agents and 3 non-identical items $A$, $B$, and $C$. The agent valuations are realized as follows:
\begin{itemize}
	\item With probability $1/2$, agent $1$ has value 0.01 for $A$, agent $2$ has value $1$ for $C$ and agent $3$  values item $C$ either $5$ (with probability 0.2) or $0.5$ (with probability 0.8). All the other values are $0$.
	\item Otherwise, agent $1$ has value 0.01 for $B$, agent $3$ has value $0.499$ for $C$ and agent $2$  values item $C$ either $2$ (with probability 0.2) or $0$ (with probability 0.8). All the other values are $0$.
\end{itemize}

Here, a welfare-optimal $\MM$  needs to adapt both the  order and  the prices. First, agent $1$ is considered with price $0$ for all items. Then, 
\begin{itemize}
	\item If agent $1$ buys item $A$,  the mechanism should visit agent $3$ with a price between $0.5$ and $5$ for item $C$ , and finally consider agent $2$ with a price smaller than $1$ for item $C$ (the price for item $B$ can be arbitrary for agents 1,2). 
	\item If agent $1$ buys $B$,  the mechanism should visit agent $2$ with a price between $0$ and $2$ for item $C$, and  finally consider agent 3 with a price smaller than 1 for item $C$ (the price for item $B$ can be arbitrary for agents 1,2). .
\end{itemize}

\paragraph{ID.}
This environment consists of 6 unit-demand agents and 2 identical items. The value of agents 1, 2, and 3 for either one of these items is drawn uniformly at random from the set $\{0,60\}$, while the value of agents 4, 5, and 6 is realized as follows:
\begin{itemize}
	\item If only agent 1 has value $60$, then agent 4's value is drawn uniformly at random from the set $\{40,0\}$, and the values of agents 5 and 6 are drawn uniformly at random  from the set $\{21,0\}$. 
	\item If only agent 2 has value $60$, then agent 5's value is drawn uniformly at random from the set $\{40,0\}$, and the values of agents 4 and 6 are drawn uniformly at random from the set $\{21,0\}$.
	\item If only agent 3 has value $60$, then agent 6's value is drawn uniformly at random from the set $\{40,0\}$, and the values of agents 4 and 5 are drawn uniformly at random from the set $\{21,0\}$. 
	\item Otherwise, the value of agents 4, 5, and 6, are all  0. 
\end{itemize}

The welfare-optimal mechanism first considers agents 1, 2, and 3, with an identical price between 0 and 60 for both items.
\begin{itemize}
	\item If only agent 1 takes an item, then, agent 4 should be visited before agents 5 and 6, with a price between 0 and 40, and then agents 5 and 6 should be visited with a price between 0 and 21.  
	\item If only agent 2 takes an item, then, agent 5 should be visited before agents 4 and 6, with a price between 0 and 40, and then agents 4 and 6 should be visited with a price between 0 and 21.  
	\item If only agent 3 takes an item, then, agent 6 should be visited before agents 4 and 5, with a price between 0 and 40, and then agents 4 and 5 should be visited with a price between 0 and 21.  
\end{itemize}

Note that this mechanism cannot be implemented by a policy that only accesses information about  the remaining agents and items.

\subsection{Part 3: Beyond Unit-demand, and Beyond Welfare Maximization.}\label{sec:other-settings}

\paragraph{Additive-across-types.}
This environment consists of 10 agents, 2 units of item $A$, and 4 units of item $B$. Agents' valuations are additive across types and unit-demand within types, meaning that an agent's value for a  bundle $x$ is given by the sum of the agent's value for item $A$, if $x$ contains at least one unit of item $A$, and the agent's value for item $B$, if $x$ contains at least one unit of item $B$. Each agent's values for each item type is  distributed as in the simple correlated setting, with correlation parameter $\delta = 0.5$.

\paragraph{Revenue Maximization.}
Here we use the same settings with 20 agents and 5 identical items we used in the simple correlated setting, with $\delta = 0.5$.

\paragraph{Max-Min Fairness}
In this environment, there are 9 unit-demand agents of which we say that one of the agents is orange, four of the agents are blue, and four of the agents are red,
and there are five black items and five white items. 
Agent values are realized as follows:
\begin{itemize}
	\item With probability 1/2: 
	\begin{itemize}
		\item The value of the orange agent for the black items is $U[0.5,1]$ and its value for the white items is 0.
		\item Blue agents' values are drawn i.i.d.~from the distribution $U[0.4,0.5]$ for black items, and  from $U[0,0.25]$ for the white items.
		\item Red agents' values are drawn i.i.d.~from $U[0.9,1]$ for the black items, and from $U[0.4,0.5]$ for the white items.
	\end{itemize}   
	\item Otherwise: (switching roles)
	\begin{itemize}
		\item The value of the orange agent for the black items is 0 and for  the white items is  $U[0.5,1]$.
		\item Red agents' values are drawn i.i.d.~from the distribution $U[0.4,0.5]$ for black items, and  from $U[0,0.25]$ for the white items.
		\item Blue agents' values are drawn i.i.d.~from $U[0.9,1]$ for the black items, and from $U[0.4,0.5]$ for the white items.
	\end{itemize}   
\end{itemize}

An optimal $\MM$  visits the orange agent first with price 0. If this agent takes a black item, the mechanism next visits blue agents, letting them take all the remaining black items. Otherwise, if the orange agent takes a white item, the mechanism should consider the red agents  next, letting them take all the black items.

The max-min welfare of such a mechanism is $\ge 0.4$, as this is the minimal welfare any agent can get. Under a static $\MM$, the max-min welfare of at most $0.25$.
If, for instance, the orange agent takes a black item and then some red agents arrive before blue agents, red agents might take more black items, which results in a blue agent taking a white item and yields a max-min welfare of at most $0.25$.

\section{Characterization Results for Personalized Static Price (PSP) Mechanisms}\label{sec:personalizedchar}

We further characterize the class of Personalized Static Price (PSP) mechanisms with results that do not appear in the main paper.
In Proposition~\ref{dynamic_order_price} we show that both adaptive prices and adaptive order are needed when maximizing welfare with identical items and independently drawn unit-demand valuations. We show that in the same setting, if we restrict ourselves to mechanisms in the PSP class, adaptive order is not needed (Proposition \ref{static_order_price}). If, however, the items are non-identical (Proposition~\ref{dynamic_order_diff}), or the valuations are correlated (Proposition~\ref{dynamic_order_corr} in the main paper),  then adaptive order might  be needed.
\begin{proposition}
	\label{static_order_price}
	For independently distributed unit-demand valuations and identical items, there exists an optimal PSP mechanism that uses static order.
\end{proposition}
\begin{proof}
	We show that for independently distributed unit-demand valuations and identical items, an optimal policy is to order the agents according to the expected value conditioned on being allocated. This implies there exists an optimal static order.
	
	For this, fix prices $\p$.
	For player $i$, let $v_i = \E\left[v\sim F_i|v> p_i \right]$, and let $q_i= \Pr[v\sim F_i> p_i]$.
	Given a set of agents $S$ and a number of items $m$, let $V(S,m)$ be the expected welfare of the optimal (possibly adaptive) ordering for set $S$ and $m$ items. We prove by induction that for any set $S$ and any $m$, it is always (weakly) better to have an agent in $\argmax_{i\in S}v_i$ first. For two buyers, 1 and 2, where $v_1>v_2$, this is true. If $m\geq 2$, it doesn't matter who goes first. If $m=1$, then the expected welfare from having 1 go first is $q_1v_1 + (1-q_1)q_2v_2$, and symmetrically, if agent 2 goes first, we have an expected welfare of $q2v_2 + (1-q_2)q_1v_1$. We have that the difference between the first and the second terms is $q_1q_2v_1-q_1q_2v_2>0$, which implies it is better to have agent 1 go first.
	
	For an arbitrary $S$ and $m$, we compare the benefit of having agent $1\in \argmax_{i\in S} v_i$ going first to an agent $j$ with $v_j< v_1$.  First, we show the claim for $m=1$. The expected welfare from having 1 go first is 
	\begin{eqnarray*}
		q_1v_1  + (1-q_1)V(S\setminus\{1\},1)\ \geq\  q_1v_1 + (1-q_1)q_jv_j + (1-q_1)(1-q_j)V(S\setminus\{1,j\},1).
	\end{eqnarray*}
    \noindent	
	The expected welfare from having agent $j$ go first is 
	\begin{eqnarray*}q_jv_j + (1-q_j)V(S\setminus\{j\},1)\ =\ q_jv_j + (1-q_j)q_1v_1 + (1-q_1)(1-q_j)V(S\setminus\{1,j\},1),\end{eqnarray*}
	where the equality follows from the induction hypothesis, where we assume it's optimal for agent 1 to go first for a smaller set.
	Subtracting the second term from the first term, we get  $q_1 q_j v_1 - q_1 q_j v_j > 0,$ implying it's strictly better for agent 1 to go first. 
	
	If $m\geq 2$, then letting agent $1$ go first, we get an expected welfare of 
	\begin{eqnarray*}q_1v_1 + q_1V(S\setminus\{1\},m-1)+ (1-q_1)V(S\setminus\{1\},m) &\geq& q_1v_1\ +\  q_1q_j(v_j+V(S\setminus\{1,j\},m-2)) \\ & &\ +\  q_1(1-q_j)V(S\setminus\{1,j\},m-1) \\& & \ +\  (1-q_1)q_j(v_j+V(S\setminus\{1,j\},m-1))\\& &\ +\  (1-q_1)(1-q_j)V(S\setminus\{1,j\},m).
	\end{eqnarray*}
	Letting agent $j$ go first results in an expected welfare of 
	\begin{eqnarray*}
		q_jv_j + q_jV(S\setminus\{j\},m-1)+ (1-q_j)V(S\setminus\{j\},m)  
		& = & q_jv_j\ +\ q_1q_j(v_1+V(S\setminus\{1,j\},m-2)) \\ & &+\ q_j(1-q_1)V(S\setminus\{1,j\},m-1) \\& & +\ (1-q_j)q_1(v_1+V(S\setminus\{1,j\},m-1))\\& &+\ (1-q_1)(1-q_j)V(S\setminus\{1,j\},m),
	\end{eqnarray*}
	where the equality follows the induction hypothesis.
	Subtracting the second term from the first, we get a difference in welfare of at least %\dcp{expected welfare?} 
	$$q_1v_1 + q_1q_jv_j + (1-q_1)q_jv_j- q_jv_j + q_1q_j v_1 + (1-q_j)q_1 v_1 = 0,$$ which implies it's weakly better to have agent 1 go first.
	
	Since it's always weakly better to have the agent with the highest value  go first, it implies there exists an optimal static ordering, where agents arrive according to descending values.
\end{proof}

\begin{proposition}
	\label{dynamic_order_diff}
	There exists a unit-demand setting with 2 non-identical items and three agents with independent valuations where the optimal \MM \,must use adaptive order.
\end{proposition}
\begin{proof}
	There are three agents: blue, red, and yellow. There are two items: red and yellow. With prob.\ 1/2, the blue agent has value 15 for the red item and 1 for the yellow item, and with prob.\ 1/2 it's the other way around. The red agent's value  for the red item is drawn uniformly from $\{12,3\}$ and its value for the yellow item is drawn uniformly from $\{8,2\}$. The yellow agent is the same as the red except the distributions are switched, so $\{12,3\}$ for  the yellow item and $\{8,2\}$ for  the red item. The welfare-optimal policy is to have the blue agent go first, with any price $p<15$ set for each of the items. One of the items will be sold. If the red item remains, visit the red agent next, with price $3<p<12$, and if this agent does not buy this item, visit the yellow agent with price $p<2$. Do the opposite if it is the yellow item that remains after the blue agent goes.
\end{proof}

\medskip

Intuitively, after the blue agent goes, we want to visit the agent who is the same color as the remaining item because they have the highest potential value. It is insufficient to use a static because we don't know for sure which of the yellow or red agents has the higher value for the item.

\end{document}